\documentclass[11pt]{article}

\usepackage{graphics}   
\usepackage{endnotes}            
\usepackage{amsfonts}            
\usepackage{amssymb}             
\usepackage{amsthm}              
\usepackage{amsmath}            
\usepackage{algorithm}           
\usepackage{algorithmic}
\usepackage[letterpaper]{geometry}
\usepackage{rotating}            

\usepackage{setspace}                
\usepackage{natbib}
\bibpunct{(}{)}{;}{a}{,}{;}

\begin{document}

\title{A Dynamic Process Reference Model for Sparse Networks with Reciprocity\thanks{This work was supported by NSF awards SES-1826589, IIS-1939237, and DMS-1361425.}
}

\author{
Carter T. Butts\thanks{Departments of Sociology, Statistics, Computer Science, and EECS and Institute for Mathematical Behavioral Sciences; University of California, Irvine; SSPA 2145, UCI, Irvine, CA 92697-5100; \texttt{buttsc@uci.edu}}
}
\date{12/30/19}
\maketitle

\begin{abstract}
Many social and other networks exhibit stable size scaling relationships, such that features such as mean degree or reciprocation rates change slowly or are approximately constant as the number of vertices increases.  Statistical network models built on top of simple Bernoulli baseline (or reference) measures often behave unrealistically in this respect, leading to the development of sparse reference models that preserve features such as mean degree scaling.  In this paper, we generalize recent work on the micro-foundations of such reference models to the case of sparse directed graphs with non-vanishing reciprocity, providing a dynamic process interpretation of the emergence of stable macroscopic behavior.\\[5pt]
\end{abstract}

\theoremstyle{plain}                        
\newtheorem{axiom}{Axiom}
\newtheorem{lemma}{Lemma}
\newtheorem{theorem}{Theorem}
\newtheorem{corollary}{Corollary}

\theoremstyle{definition}                 
\newtheorem{definition}{Definition}
\newtheorem{hypothesis}{Hypothesis}
\newtheorem{conjecture}{Conjecture}
\newtheorem{example}{Example}

\theoremstyle{remark}                    
\newtheorem{remark}{Remark}


\section{Introduction}

Recent advances in stochastic models for complex networks have provided an increasingly flexible and powerful ``toolkit'' for modeling social and other networks, motivating a growing interest in identifying model classes that generalize well across settings \citep[see e.g.][]{vanduijn.et.al:sn:1999,goodreau.et.al:d:2009,huitsing.et.al:sn:2012,mcfarland.et.al:asr:2014,schweinberger.et.al:ss:2019}, and in linking models for cross-sectional network structure with the generative processes that give rise to them \citep[e.g.][]{snijders:sm:2001,skvoretz.et.al:sn:2004,butts:jms:2019}.  Both issues become increasingly important when attempting to create models that generalize to networks of varying size, as interactions within large social, biological, or physical systems are generally impeded by geographical, physical, institutional, or other barriers whose influence on network structure can be profound \citep{mcpherson.et.al:ars:2001}.  Even when such barriers are unobserved - and where explicitly modeling them is infeasible - it is often necessary to account for their tacit influence in order to obtain reasonable model behavior.

Perhaps the most well-known manifestation of this phenomenon is in the scaling of mean degree with vertex set size (henceforth, $N$).  While a simple baseline model in which each pair of nodes has some constant probability of being tied would suggest that mean degree should scale linearly in $N$, in most systems of interest mean degree tends to be either approximately constant or at best to scale sublinearly with network size.  For example, Figure~\ref{f_denrecip} shows the scaling of density (red dots) versus $N$ for friendship nominations in schools from wave 1 of the public use sample of the AddHealth study \citep{harris.et.al:web:2009} and radio calls from teams of responders in the 2001 World Trade Center disaster \citep{butts.et.al:jms:2007}.  While both sets of networks reflect very different types of social relations, in both cases we see an approximately $1/N$ relationship between density and size (95\% CIs of power law exponents (-1.12, -1.01) for AddHealth, (-1.07,-0.76) for WTC), indicating little variation in mean degree over a large range of network sizes.  Such scaling is typical of social networks, as exemplified by the observation that, while the human population has grown over an order of magnitude since the mid-18th century \citep{caselli.et.al:bk:2006}, we have not seen a 10-fold increase in the average number of personal ties.  (Nor do the residents of New York City ($N>8.6\times 10^6$) have over 10,000 times as many friends as residents of Colerain, NC ($N=183$).)  To correctly account for this behavior is the \emph{mean degree scaling problem,} for which multiple solutions have been proposed.  For instance, limits on the capacity to sustain ties (due e.g. to time and effort \citep{mayhew.levinger:ajs:1976} or cognitive capacity \citep{dunbar:bk:1997}) would bound the maximum degree, and hence lead to an asymptotically constant mean degree for sufficiently large systems.  Alternatively, constraints on interaction opportunities associated with the greater geographical, institutional, economic, or even cultural dispersion (including intentionally designed constraints \citep{galbraith:bk:1977}) typical of larger social systems have also been suggested as sources of sparsity \citep[e.g.][]{blau:ssr:1972,blau:ajs:1977,carley:asr:1991,mcpherson:icc:2004,butts.et.al:sn:2012}.  While both types of factors are possible (and no doubt present in different settings), the type of underlying mechanism has important implications for social structure.  For instance, \citet{butts:jms:2019} demonstrates that maximum degree constraints inevitably lead to degree saturation in the large-$N$ limit (i.e., nearly all individuals having as many ties as possible at all times), casting doubt on the viability of this class of mechanisms as an explanation for mean degree scaling in typical social networks.  By contrast, \citet{butts:jms:2019} also shows that a simple model of latent social interaction - the \emph{contact formation process} (CFP) - in which tie formation is constrained to transient co-location within social settings \citep[i.e., \emph{foci}, per][]{feld:ajs:1981} is able to account for constant (or, under alternative assumptions, non-constant) mean degree scaling without saturation effects.  The core intuition of this model is that, in many networks, tie formation requires some form of co-location for interaction to take place (whether it be a physical location, an organizational co-involvement, or even a region within a topic space).  Movement of individuals between settings may be very rapid compared to the timescale of network dynamics, such that the migration dynamics are ``blurred out'' on the network timescale.  However, this unobserved process still leaves its mark on the network structure, in this case by altering the expected degree.  Where the number of social settings scales in proportion to population size, the expected degree remains constant.

\begin{figure}[h]
  \centering
    \includegraphics[width=0.45\textwidth]{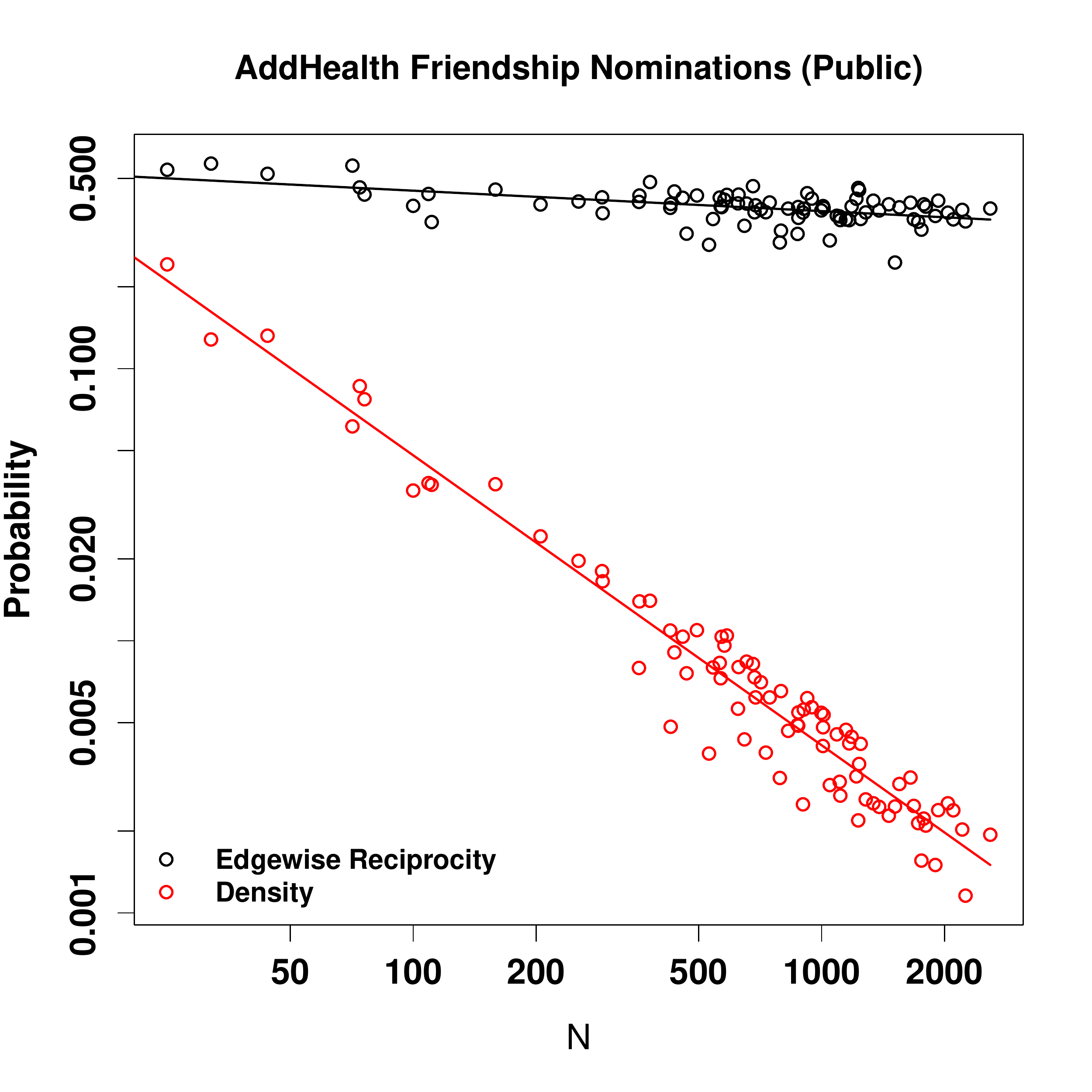}\includegraphics[width=0.45\textwidth]{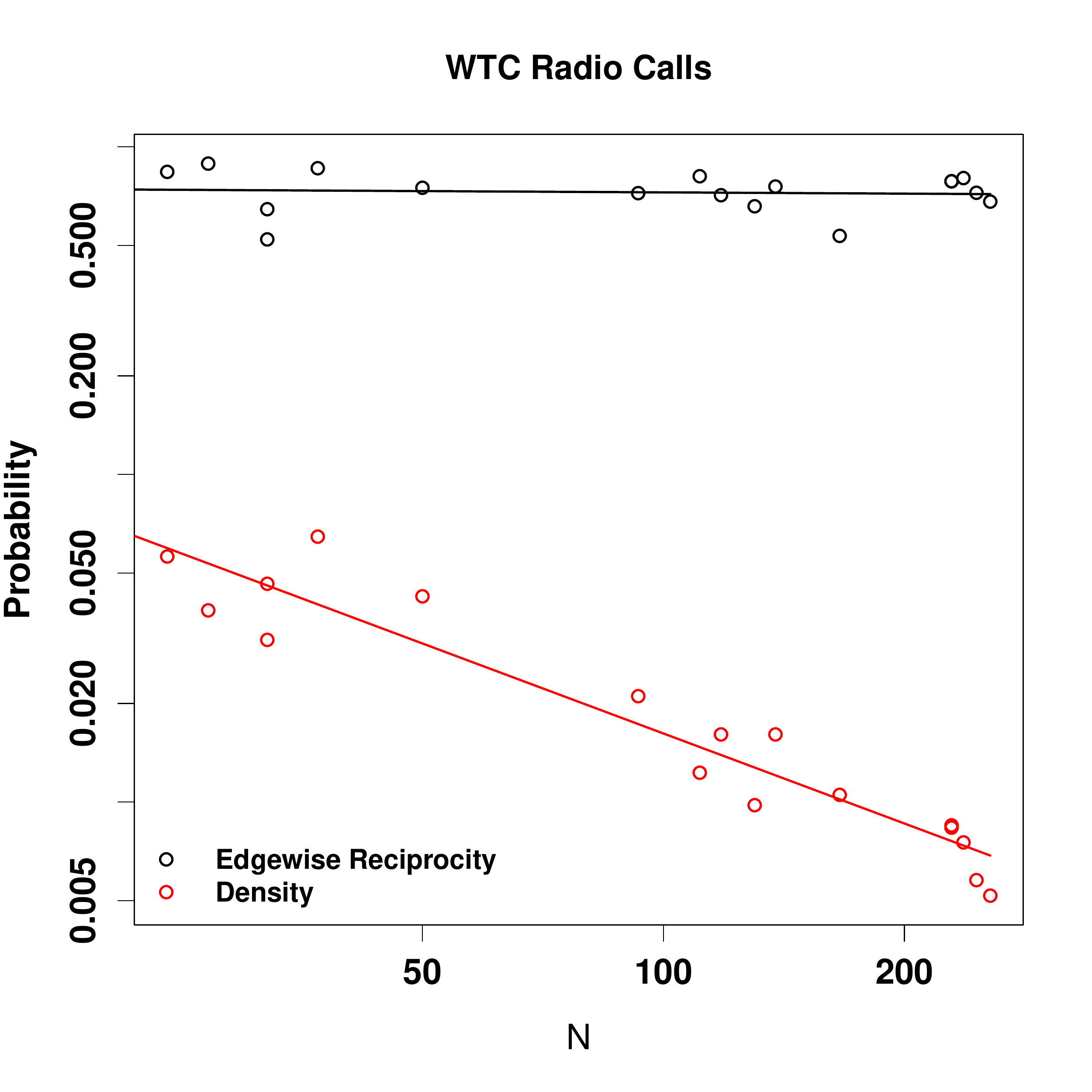}
    \caption{Density and reciprocity versus $N$ for AddHealth friendships (left) and WTC radio calls (right); dots are observed networks, with lines indicating OLS regressions for each quantity on log-log scale.  In both cases, density falls as approximately $1/N$ (implying roughly constant mean degree) while edgewise reciprocity is approximately constant over several orders of magnitude variation in network size.\label{f_denrecip}}
\end{figure}

Constant mean degree scaling raises other questions, however.  The $1/N$ density scaling corresponding to constant mean degree implies that, in the directed case, reciprocating edges should become increasingly rare as population size increases.  While this may be true of some networks, it is often counter to what is observed.  For instance, we can see from Figure~\ref{f_denrecip} that the probability edge reciprocation (the edgewise reciprocity) is nearly constant for both the AddHealth and WTC radio networks, despite sharply declining density.  This is difficult to explain in terms of \emph{ad hoc} mutuality effects, as it would require that such effects be distributed in precisely the right way to balance out density decline; put another way, dyads would have to ``know'' the size of the network in which they were embedded, in order to cancel out its effects.  While it is possible to construct models with mutuality effects that are by fiat set to become stronger with $N$ to maintain constant reciprocity \citep[as was done by][]{krivitsky.kolaczyk:ss:2015}, this phenomenological fix does not provide a mechanistic explanation for where constant reciprocity comes from, or how it can co-exist with declining density.  

In the remainder of this paper, we provide one mechanistic account for the coincidence of mean degree scaling and constant reciprocity, using a simple extension of the contact formation process.  We characterize the graph distribution arising from this extended CFP, and show how it can be employed as a reference model for exponential family random graph models (to which other effects can be added to obtain more complex models); this provides a mechanistic interpretation of the offset procedure suggested by \citet{krivitsky.kolaczyk:ss:2015} for preserving reciprocity when extrapolating ERGMs to networks of increasing size.  We explore the degree of timescale separation required for the underlying migration process to have the desired structural properties, and also comment on generalizations to alternative assumptions regarding the scaling of social settings with population.

\subsection{Exponential Family Random Graph Models}

Exponential family random graph models provide a very general framework for writing probability distributions on graph sets, and we will employ them here.  Given a random graph $G$ with support $\mathcal{G}$, the probability mass function (pmf) of $G$ in ERGM form is given by
\begin{equation}
\Pr(G=g|\theta,X) = \frac{\exp(\theta^T t(g,X)) h(g,X)}{\sum_{g' \in \mathcal{G}}\exp(\theta^T t(g',X)) h(g',X)}, \label{e_ergm}
\end{equation}
where $\theta \in \mathbb{R}^p$ is a vector of parameters, $t:\mathcal{G},X \mapsto \mathbb{R}^p$ is a vector of sufficient statistics encoding graph properties, $X$ is a covariate set, and $h$ is the reference measure with respect to which the graph distribution is defined.  This form is quite general, and indeed any distribution on finite $\mathcal{G}$ can be written in this manner (albeit not necessarily parsimoniously).  Considerable statistical theory exists regarding inference for $\theta$ given observations of $G$ \citep[see e.g.][]{hunter.handcock:jcgs:2006,hunter.et.al:jcgs:2012,schweinberger.stewart:as:2019,schweinberger.et.al:ss:2019}, and various algorithms for exact \citep{butts:jms:2018} or approximate \citep{snijders:joss:2002,morris.et.al:jss:2008,butts:jms:2015} simulation of ERGM draws have been developed.  Here, however, our focus is on using the ERGM form to describe the equilibrium behavior of the extended contact formation process (thus providing a mechanistically motivated reference model from which more complex models can be constructed).

\subsection{Baseline Models and ERGM Reference Measures}

The role of $h$ in Equation~\ref{e_ergm} can be appreciated by taking $\theta=0$, immediately leading to $\Pr(G=g|\theta=0,X) \propto h(g,X)$.  We can hence view $h$ (together with $\mathcal{G}$) as defining the \emph{baseline model} associated with a particular ERGM family, in the sense that any other distribution in the same family can be viewed as a reweighted (exponentially tilted) or biased version of the baseline model.  Most commonly, $h(g,X)$ is taken to be a constant (i.e., the counting measure), leading to the uniform distribution on $\mathcal{G}$ (a ``true'' baseline model in the original sense of \citet{mayhew:jms:1984a}); where $\mathcal{G}$ is the set of all graphs or digraphs on $N$ vertices, this distribution can also be interpreted in terms of a process in which each potential edge occurs independently with probability 0.5.

From an interpretative standpoint, it can be useful to broaden the above by including with $h$ a minimal set of terms (and associated parameters) that are thought to collectively define the baseline distribution on which the other terms act.  Specifically, we may re-write Equation~\ref{e_ergm} in this ``separated'' form as follows:
\begin{equation}
\Pr(G=g|\phi,\theta,X) = \frac{\exp(\psi^T t_r(g,X) + \theta^T t(g,X)) h(g,X)}{\sum_{g' \in \mathcal{G}}\exp(\psi^T t_r(g',X) + \theta^T t(g',X)) h(g',X)}, \label{e_refergm}
\end{equation}
where we have identified $\psi$ as reference parameters and $t_r$ as the associated statistics.  Taking $\theta=0$ then leads to what we may call the \emph{reference model} (or reference model family),
\begin{equation}
\Pr(G=g|\phi,\theta=0,X) = \frac{\exp(\psi^T t_r(g,X)) h(g,X)}{\sum_{g' \in \mathcal{G}}\exp(\psi^T t_r(g',X) ) h(g',X)}. \label{e_refmod}
\end{equation}

The logic of such a parameterization can be appreciated by noting that certain terms (in addition to choices of reference measure) are immediately motivated by one's choice of baseline process.  For instance, consider a baseline process in which, for a directed relation on an $N$-person group, each person nominates each other person at random with some constant probability $p$ (not necessarily equal to 0.5).  This corresponds to a reference model with $h(g,X)\propto 1$, $t_r(g,X)=t_e(g)$ (where $t_e(g)$ is the number of edges in $g$), and $\psi=\mathrm{logit}(p)$, i.e. a homogeneous Bernoulli graph with expected density $p$.  Such a reference model is a natural starting point for e.g., networks within small, closed group settings in which all individuals are aware of one another and (\emph{a priori}) anyone could plausibly nominate anyone else.  Additional terms added to this model may then be interpreted in terms of social forces that act to bias the underlying nomination process.

In addition to motivating certain terms, the baseline process may also motivate particular interpretations of the reference parameters.  For instance, the contact formation process of \citet{butts:jms:2019} leads to a reference model with $h(g,X)=N^{-t_e(g)}$, $t_r(g,X)=t_e(g)$, and $\psi=\log \tfrac{r_f P}{r_\ell}$, where $r_f$ is the tie formation rate, $r_\ell$ is the tie loss rate, and $P$ is the mean number of persons per focus.  Although $r_f$, $r_\ell$, and $P$ are typically unknown in inferential settings ($\psi$ estimated via a free parameter), the process interpretation does allow one to determine the hypothetical impact of e.g. increasing or decreasing the formation or loss rates, or changing the population density of foci on a specified model.  (For instance, doubling the underlying tie loss rate is equivalent to reducing the edge parameter by $\log 2$, the consequences of which in a more complex model can be explored via simulation.)  

Finally, we note that the baseline process may clarify which reference model features should be thought of as belonging to the reference measure, and which should be viewed as belonging to the exponentiated linear predictor.  This is less obvious than it may seem, since it is always possible to fold $h$ into the exponentiated portion of the model by use of offset parameters (i.e., parameters whose values are fixed).  For instance, the sparse graph reference model originally proposed by \citet{krivitsky.et.al:statm:2011} was described in terms of a counting measure ERGM with a $-\log N$ offset to the edge parameter.  The same approach was used by \citet{butts.almquist:jms:2015} for their generalization to ERGMs with power law mean degree scaling, and by \citet{krivitsky.kolaczyk:ss:2015} for their reciprocity and transitivity preserving models.  This approach to specification has the virtue of suggesting a simple implementation \emph{vis a vis} conventional software tools, and from a purely phenomenological standpoint is perfectly adequate for most current social science applications.  However, as noted above, this comes at some interpretational cost.  In particular, the reference measure has an important interpretation in terms of the (exponentiated) entropy of the graph microstate \citep[per][]{jaynes:bk:1983}, which is distinct from the effects of model terms (which are analogous to energetic effects, in the sense of internal (as opposed to free) energy).  In more intuitive terms, the underlying process that produces the network may contain hidden degrees of freedom that \emph{create more ways to produce some graphs than others}, and this is distinct from the action of social forces that bias this underlying process towards one or another outcome.  For instance, use of the counting measure implies that there is in essence the same number of ``ways'' of producing one graph as any other, and hence that all inequalities in graph probability stem from the action of social forces.  By contrast, the reference model arising from the CFP implies that there are more ways to obtain a sparse graph than a dense one (in that case, because edge formation requires vertices to share foci, and there are fewer ways for vertices to be co-located than not), and this is an inherent contributor to sparsity apart from the action of social forces.  Since such inherent biases can be expected to turn up in any network generated by the same class of processes (while social forces may vary), it is potentially useful to identify them.  Further, such identification can be critical in non-social settings, where the distinction between entropic and energetic contributors to structure is well-defined and consequential.  In particular, entropic contributions are unaffected by temperature, while changing the temperature in a physical system has the effect of rescaling $\phi$ and $\theta$; ERGMs intended for use at multiple temperatures must hence correctly distinguish between elements of $h$ and offset terms.  (Such an approach has been used e.g. by \citet{grazioli.et.al:jpcB:2019} to apply ERGMs to the modeling of protein aggregation.)  Whether or not the energetic cost of \emph{social} structure in literal terms is ultimately quantifiable (and useful) - as argued e.g. by \citet{mayhew.et.al:sf:1995} - remains an open question, but isolating the entropic drivers of structure from those arising from other sources is without doubt an important theoretical objective.

\section{A Contact Formation Process with Reciprocation}

As noted, our proximate goal is to construct a simple but plausible ``baseline'' process that can account for the co-existence of constant mean degree scaling and constant reciprocity.  We propose a simple extension of the contact formation process proposed by \citet{butts:jms:2019}, which can be informally motivated as follows.  We consider our social system to consist of a set of actors that are embedded within a set of social, institutional, or geographical ``locations'' that are fixed on the time scale of network evolution and between which individuals can readily migrate.  Following \citet{feld:ajs:1981}, we refer to these generalized locations as \emph{foci}; while we make no assumptions regarding their substantive interpretation, foci are assumed to play a critical role in tie formation.  In particular, ties within null dyads can only be formed when both members of the dyads in question reside within the same focus.  Applying this restriction to all tie formation (not only to the first tie within a dyad) results in a directed version of the CFP proposed by \citet{butts:jms:2019}.  Here, however, we assume that incoming ties themselves constitute a context for tie formation, and hence allow reciprocating ties to be formed irrespective of the foci in which the respective vertices reside (the distinguishing feature of the CFPR).  Regardless of how they are formed, ties have a constant hazard of dissolution, and vertices are assumed to migrate between foci at random (carrying their ties with them).

Formally, we define the CFPR as follows.  Assume a system of $N$ vertices, $\mathcal{V}$, each of which at any given moment resides within one of $M$ foci.  Let $\mathcal{G}_{\mathcal{V}}$ be the set of all digraphs on $\mathcal{V}$, and let $\mathcal{F}_\mathcal{V}=\{1,\ldots,M\}^N$ be the set of all possible assignments of vertices to foci (bearing in mind that some foci may be empty).  The CFPR is a continuous time Markov process on state space $\mathcal{G}_\mathcal{V} \times \mathcal{F}_\mathcal{V}$, whose permissible transitions consist of: (1) adding a directed edge between non-adjacent ordered pair $(i,j)$ when either (i) $i$ and $j$ reside in the same focus, or (ii) the $(j,i)$ edge exists; (2) removing an $(i,j)$ edge; and (3) moving a single vertex from its current focus to another.  The hazards for these transitions are determined as follows.  Every directed dyad at risk for formation has a constant hazard $r_f$ of forming an edge (with the risk set being the set of non-adjacent $(i,j)$ ordered vertex pairs such that either $j$ is adjacent to $i$ or $i$ and $j$ occupy the same focus).  Likewise, every currently existing $(i,j)$ edge has a constant hazard $r_\ell$ of dissolution (with this being invariant to location or network structure).  Finally, every vertex has a constant hazard $r_m/(M-1)$ of migrating from its current focus to each other focus (yielding a total per-vertex migration hazard of $r_m$, if we imagine that with probability $1/M$ a migrating vertex elects to remain within the same focus).

From these conditions, we may immediately observe that, so long as all hazards are positive and finite, the CFPR has neither transient nor absorbing states.  Its state space is finite and connected (obviously, we may obtain any combination of graph structure and focus assignments from any other by a series of movements and edge additions or deletions), and the Markov chain formed by the above transition rules is non-periodic.  It follows then that the CFPR is ergodic, with a stationary (or equilibrium) distribution to which the system converges.  The exact properties of this distribution depend upon the associated rates, with the migration rate being of particular importance.  As with development of \citet{butts:jms:2019}, we are here interested in the \emph{fast mixing regime}, in which $r_m \gg r_f,r_\ell$.  In this limit, migration is much faster than tie formation or dissolution, with foci representing transient sites of interaction (e.g., meetings or meeting places) rather than long-term contexts in which individuals are embedded.  While the assignment of individuals to foci is ``blurred out'' in this regime, the underlying dynamics nevertheless affect network structure, as we show below.

\paragraph{Event Representation:}  While the above is one characterization of the CFPR, other equivalent characterizations can be useful for specific purposes.  In particular, in showing the asymptotic independence of dyads in the fast mixing regime we will make use of an event-based representation of the CFPR.  As before, we define the CFPR as a continuous time process on state space $\mathcal{G}_\mathcal{V} \times \mathcal{F}_\mathcal{V}$, letting $G^{(t)} \in \mathcal{G}_\mathcal{V}$ being the state of the network at time $t$, and $F^{(t)}\in\mathcal{F}_\mathcal{V}$ the corresponding vector of focus assignments.  We then associate with each ordered vertex pair $(i,j)$ a set of \emph{formation events} and a set of \emph{dissolution events}, and each vertex/focus pair $(i,k)$ a set of \emph{migration events}, all of which occur as Poisson processes.  Specifically, the dissolution events for pair $(i,j)$ and migration events for pair $(i,k)$ occur as homogeneous Poisson processes with respective constant hazards $r_\ell$ and $r_m/M$, for all ordered vertex and vertex/focus pairs (respectively).  Formation events for vertex pair $(i,j)$ occur as an inhomogeneous Poisson process with piecewise constant hazard $r_f$ when $F^{(t)}_i=F^{(t)}_j$ or when $(j,i)\in G^{(t)}$, and otherwise 0.  Under this construction, we may recover the system state at time $t$ by examining its event history: $(i,j)\in G^{(t)}$ if and only if the most recent $(i,j)$ event was a formation event and $F^{(t)}_i=k$ if and only if the most recent migration event involving vertex $i$ was an $(i,k)$ event.  Although they can be viewed as a pure mathematical contrivance, we may also think of formation events in this construction as representing hypothetical opportunities for tie formation to occur; such an event creates the corresponding tie if it is not already present, and otherwise has no effect.  Similarly, dissolution events terminate edges that are present, but otherwise do nothing.  Expressing the CFPR in terms of these underlying events facilitates certain results, as shown below.

\subsection{Equilibrium Behavior of the CFPR} \label{sec_eq}

We now consider the behavior of the CFPR in equilibrium (i.e., when observed at a random time), in the fast-mixing regime for which $r_m \gg r_f,r_\ell$.  In the development below, it will often be convenient to work with the mean number of vertices per focus, which we denote by $P=N/M$; we are particularly interested in the behavior of systems at constant population density, i.e. for which $P$ is constant in $N$.  Cases in which $P$ is not constant (i.e., $M$ does not scale linearly in $N$) are discussed in section~\ref{sec_varm}.

\subsubsection{Expected Dyad Census}

We begin by determining the expected dyad census at equilibrium, in the fast-mixing regime.  Denote the expected numbers of mutual, asymmetric, and null dyads in the network at a random time by $D_m$, $D_a$, and $D_n$ (respectively).  From the definition of the CFPR, we can immediately observe that edges are lost at constant rate $r_\ell$; thus, mutual dyads must convert to asymmetric dyads at rate $2 r_\ell$, and asymmetric dyads must likewise convert to null dyads at rate $r_\ell$.  Since reciprocation is always permitted, we can immediately see that asymmetric dyads convert to mutual dyads at rate $r_f$, irrespective of vertex location.  Null conversion, however, is more subtle.  For pairs within the same focus, nulls convert to asymmetrics at rate $2r_f$, with the conversion rate otherwise being 0.  The mean rate of conversion hence depends upon the latent migration process.  Since migration is uniform, we observe that in equilibrium all locations are chosen at random; thus, the chance of two vertices within a dyad occupying the same focus is $1/M$, and the net conversion rate is hence $2/M r_f$.

From these conversion rates, we can obtain the expected dyad census.  In equilibrium, the expected gains and losses of mutual dyads must be equal, and hence we can solve for the expected number of asymmetrics as a function of the expected number of mutuals:
\begin{align}
  D_a r_f &= D_m 2 r_\ell \nonumber \\
      D_a &= 2 \frac{r_\ell}{r_f} D_m. \label{eq:DaasDm}
\end{align}
Applying the same logic to the gain/loss rates for asymmetric dyads, and substituting from equation~\ref{eq:DaasDm} also allows us to solve for the expected number of nulls as a function of the expected number of mutuals:
\begin{align}
  D_m 2 r_\ell + \frac{2}{M} r_f D_n &= D_a r_f + D_a r_\ell \nonumber \\
                        &= 2 \frac{r_\ell}{r_f} D_m (r_f + r_\ell) \nonumber \\
            \frac{2}{M} r_f D_n &= 2 \frac{r_\ell}{r_f} D_m (r_f + r_\ell) - D_m 2 r_\ell \nonumber \\
                        &= 2 r_\ell \left[ \left(1+\frac{r_\ell}{r_f}\right) - 1 \right] D_m \nonumber \\
                        &= 2 \frac{r_\ell^2}{r_f} D_m \nonumber \\
                     D_n &= M \left(\frac{r_\ell}{r_f}\right)^2 D_m, \label{eq:DnasDm}
\end{align}
where we have used the fact that the expected conversion rate from nulls to asymmetrics in the fast-migration regime is $2/M r_f$.

Since the total number of dyads, $D$, is fixed, we may now combine equations \ref{eq:DaasDm} and \ref{eq:DnasDm} to find the expected number of mutuals in terms of the CFPR parameters (and thereby the rest of the dyad census).  Specifically,
\begin{align*}
  D &= D_a + D_m + D_n \\
    &= 2 \frac{r_\ell}{r_f} D_m + D_m + M \left(\frac{r_\ell}{r_f}\right)^2 D_m \\
    &= D_m \left[ 1 + 2 \frac{r_\ell}{r_f} + M \left(\frac{r_\ell}{r_f}\right)^2 \right],
\end{align*}
and hence
\begin{align}
 D_m &= \frac{D}{ 1 + 2 \frac{r_\ell}{r_f} + M \left(\frac{r_\ell}{r_f}\right)^2 }, \label{eq:Dm}
\end{align}
\begin{align}
 D_a &= 2 \frac{r_\ell}{r_f} \frac{D }{ 1 + 2 \frac{r_\ell}{r_f} + M \left(\frac{r_\ell}{r_f}\right)^2 } \nonumber \\
     &= \frac{2 D}{ \frac{r_f}{r_\ell} + 2 + M \frac{r_\ell}{r_f} }, \label{eq:Da}
\end{align}
and
\begin{align}
 D_n &= M \left(\frac{r_\ell}{r_f}\right)^2 \frac{ D }{ 1 + 2 \frac{r_\ell}{r_f} + M \left(\frac{r_\ell}{r_f}\right)^2 } \nonumber \\
     &= \frac{M D}{ \left(\frac{r_f}{r_\ell}\right)^2 + 2 \frac{r_f}{r_\ell} + M }. \label{eq:Dn}
\end{align}

\subsubsection{Expected Degree and Reciprocity}

From the dyad census, we can determine how participation in dyad types scales with network size.  For instance, we can see from equation~\ref{eq:Dm} that the expected number of mutuals per vertex is
\begin{align}
\frac{D }{ 1 + 2 \frac{r_\ell}{r_f} + M \left(\frac{r_\ell}{r_f}\right)^2 } \frac{1}{N}  &= \frac{N(N-1)/2}{ 1 + 2 \frac{r_\ell}{r_f} + \frac{N}{P} \left(\frac{r_\ell}{r_f}\right)^2 } \frac{1}{N} \nonumber \\
 &= \frac{N-1}{ 2 + 4\frac{r_\ell}{r_f} + 2 \frac{N}{P} \left(\frac{r_\ell}{r_f}\right)^2 }\nonumber  \\
 &\xrightarrow[N\to\infty]{} \frac{N}{ 2 \frac{N}{P} \left(\frac{r_\ell}{r_f}\right)^2 }\\
 &= \frac{1}{2} P \left(\frac{r_f}{r_\ell}\right)^2,
\end{align}
showing that the expected number of mutual relationships for an arbitrary vertex is asymptotically constant in the limit of network size.  Likewise for the expected number of asymmetrics per vertex (using equation~\ref{eq:Da}):
\begin{align}
\frac{2 D }{  \frac{r_f}{r_\ell} + 2 + M\frac{r_\ell}{r_f} }  \frac{1}{N} &= \frac{N(N-1)}{2} \frac{ 2 }{  \frac{r_f}{r_\ell} + 2 + M\frac{r_\ell}{r_f} }  \frac{1}{N} \nonumber \\
 &= \frac{N-1}{ \frac{r_f}{r_\ell} + 2 + \frac{N}{P} \frac{r_\ell}{r_f} } \nonumber \\
 &\xrightarrow[N\to\infty]{} \frac{N}{ \frac{N}{P} \frac{r_\ell}{r_f}} \\
 &= P \frac{r_f}{r_\ell},
\end{align}
which is also constant in $N$.  Since the expected in- and outdegree must be related to the expected dyad census by $\bar{d}=(2D_m+D_a)/N$ (using standard graph identities), it follows that the mean degree is asymptotically
\begin{align}
\lim_{N\to\infty}\bar{d} &=  P \left(\frac{r_f}{r_\ell}\right)^2 + P \frac{r_f}{r_\ell} \nonumber \\
        &= P \frac{r_f}{r_\ell} \left( \frac{r_f}{r_\ell} + 1\right), \label{e_meandeg}
\end{align}
which does not depend on $N$.  Likewise, we can obtain the limiting probability that a randomly chosen edge will be reciprocated (i.e., edgewise reciprocity) from
\begin{align}
\frac{D_m}{D_m + \frac{D_a}{2}} &= \frac{1}{1 + \frac{1}{2} \frac{D_a}{D_m}}  \nonumber \\
 &=  \left[ 1 + \frac{1}{2} \frac{2 D }{ \frac{r_f}{r_\ell} + 2 + M \frac{r_\ell}{r_f} }  
    \frac{ 1 + 2 \frac{r_\ell}{r_f} + M \left(\frac{r_\ell}{r_f}\right)^2 }{D} \right]^{-1} \nonumber \\
 &=  \left[1 + \frac{ 1 + 2 \frac{r_\ell}{r_f} + M \left(\frac{r_\ell}{r_f}\right)^2 }{ \frac{r_f}{r_\ell} + 2 + M \frac{r_\ell}{r_f} } \right]^{-1} \nonumber \\
 &=  \left[1 + \frac{ 1 + 2 \frac{r_\ell}{r_f} + \frac{N}{P} \left(\frac{r_\ell}{r_f}\right)^2 }{ \frac{r_f}{r_\ell} + 2 + \frac{N}{P} \frac{r_\ell}{r_f} } \right]^{-1} \nonumber \\
 &\xrightarrow[N\to\infty]{} \left[1 + \frac{ \frac{N}{P} \left(\frac{r_\ell}{r_f}\right)^2 }{ \frac{N}{P} \frac{r_\ell}{r_f} } \right]^{-1} \nonumber \\
 &= \left[1 + \frac{r_\ell}{r_f} \right]^{-1},
\end{align}
which is also constant in $N$.  (Note that we have in the first step implicitly equated an expectation of ratios with a ratio of expectations, which can be done here because of the concentration of the dyad census in the large-$N$ limit.  We verify convergence for realistic values of $N$ by simulation in section~\ref{sec_timescale}.)

\subsubsection{Asymptotic Independence of Dyads}

Having shown the limiting behavior of density and reciprocity under the CFPR, we now show that these fully characterize its behavior in the fast-mixing regime; that is, we show conditional independence of dyads in the limit as $r_m\gg r_f,r_\ell$.  Our development here closely follows that of \citet{butts:jms:2019} for the CFP, with adjustments for the directed case.  In particular, we begin by generalizing a result of \citet{butts:jms:2019} regarding the conditional distribution of vertex co-residence time, which for convenience we state as a lemma:

\begin{lemma} \label{l_ctime}
Let $G^{(0)}$ be a random-time realization of a graph arising from a contact formation process with mixing rate $r_m$, and let $i,j$ be vertices of $G$.  Define $[a,b]$, $a<b\le 0$ to be an interval prior to realization time, and let $C^b_a$ be the total time within $[a,b]$ during which $i$ and $j$ occupy the same focus.  Then, if $\Pr(G^{(0)}=g|C^b_a=c)>0$ for all $g,c$, then $C_a^b|G^{(0)}=g$ converges in probability to $(b-a)/M$ as $r_m \to \infty.$
\end{lemma}
\begin{proof}
Without loss of generality, we fix the location of $i$ as a reference, and consider the total length of time during the interval $[a,b]$ in which $j$ is in the reference focus.  Since migration events occur as homogeneous Poisson processes with rate $r_m$ for each vertex, migration events for vertex $j$ relative to vertex $i$'s position occur as a homogeneous Poisson process with rate $2r_m$ (migration of the reference position being equivalent to migration of $j$).  Since each migration event selects the destination focus at random from the $M$ that are available, the probability that such an event will result in $j$ occupying the reference focus is always $1/M$ (regardless of $j$'s starting position).  During a fixed period of duration $t=b-a$, it follows that the number of events placing $j$ at the reference position will be distributed as $\mathrm{Pois}(2 t r_m/M)$.  Each such event will be followed by a co-residence period that lasts until the next migration event (or the end of the time interval); the sum of these period lengths is by definition $C_a^b$.  (Note that migration events leaving co-located vertices positions unchanged can be viewed as ``back-to-back'' co-residence intervals, and do not harm our analysis.)  As migration events follow a homogeneous Poisson process, it immediately follows that the co-residence periods are (momentarily setting aside truncation) distributed as iid $\mathrm{Exp}(2 r_m)$, and their total length is distributed as $\mathrm{Gamma}(K,2 r_m)$ (with $K$ being the number of intervals).  This implies that the total length of co-residence is a Poisson mixture of gamma deviates.  We further observe that, in the limit as $r_m \to \infty$, the length of any given co-residence interval approaches 0 almost surely, and hence the impact of truncation (which can affect only one such interval) on the total length of co-residence must also vanish as the migration rate increases.  Thus, we are justified in equating the limiting distribution of $C_a^b$ with the sum of untruncated co-residence intervals.  The limiting expectation of $C_a^b$ can then be obtained from the Poisson mixture,
\begin{align*}
\mathbf{E}C_a^b &\to \sum_{k=0}^\infty \mathrm{Pois}(k|2 t r_m/M) \frac{k}{ 2r_m}\\
&=\frac{t}{M},
\end{align*}
where we have used the fact that the $k$th mixture component has expectation $k/(2 r_M)$.  As the $k$th component has variance $k/(2 r_m)^2$, we may also employ standard properties of discrete mixtures to obtain the limiting variance,
\begin{align*}
\mathrm{Var}(C_a^b) &\to \sum_{k=0}^\infty \mathrm{Pois}(k|2 t r_m/M) \left[\left(\frac{k}{2 r_m}-\frac{t}{M}\right)^2+\frac{k}{(2r_m)^2}\right]\\
&= \sum_{k=0}^\infty \frac{(2 t r_m/M)^2 \exp(-2 t r_m/M)}{k!} \frac{k+(k-2 t r_m/M)^2}{4 r_m^2}\\
&=\frac{t}{r_m M}.
\end{align*}
Since $\mathrm{Var}(C_a^b) \to 0$ as $r_m \to \infty$, $C_a^b \to t/M$ in mean square (and hence in probability) in the fast-migration limit.

Now we consider the implications of conditioning on $G$.  By Bayes's theorem,
\begin{equation*}
p(C_a^b=c|G^{(0)}=g) \propto \Pr(G^{(0)}=g|C_a^b=c) p(C_a^b=c),
\end{equation*}
with $p(C_a^b)$ being the marginal probability density of $C_a^b$.  From the above, however, $p(C_a^b=c) \to 0$ for all $c\neq t/M$ as $r_m \to \infty$, and since by assumption $\Pr(G^{(0)}=g|C_a^b=c) \in (0,1)$ for all $g,c$ it follows that $C_a^b|G^{(0)}=g$ converges in probability to a degenerate distribution centered at $(b-a)/M$ in the limit of increasing migration rate.
\end{proof} 

The central implication of Lemma~\ref{l_ctime} is that, under extremely broad conditions, sufficiently fast migration rates will remove any information carried by network structure regarding vertex position arbitrarily quickly.  This leads to edgewise independence under the CFP, and can be employed to show independence of dyads under the CFPR as follows.

\begin{theorem} \label{t_indep}
Let $G^{(0)}$ be a random-time realization of a graph arising under a contact formation process with reciprocity, with adjacency matrix $Y$, and let $(i,j),(k,\ell)$ be vertex pairs within $G$ such that $(i,j)\neq (k,\ell)$ and $(i,j)\neq (\ell,k)$.  Then $Y_{ij} \perp Y_{k \ell}$, in the limit as $r_m\to\infty$.
\end{theorem}
\begin{proof}
First, observe that the $i,j$ dyad has four states (which we will here represent in terms of $(Y_{ij},Y_{ji})$ pairs as (0,0), (0,1), (1,0), and (1,1)).  We prove asymptotic independence of the $i,j$ dyad from the $(k,\ell)$ edge variable by showing that the transition rates between $i,j$ dyad states become invariant with respect to the current state of $Y_{k\ell}$ as $r_m \to \infty$, for any prior time period.  

We begin by noting that all but two transitions rates trivially satisfy this invariance: from the definition of the CFPR, it immediately follows that $(0,1) \to (1,1)$ and $(1,0) \to (1,1)$ occur with fixed rate $r_f$, and that the four transitions $(0,1) \to (0,0)$, $(1,0) \to (0,0)$, $(1,1) \to (1,0)$, and $(1,1) \to (0,1)$ occur with fixed rate $r_\ell$.  This leaves us with $(0,0) \to (1,0)$ and $(0,0) \to (0,1)$.  Consider the former.  The probability of such a transition occuring within some arbitrary interval $[a,b]$ with $a<b\le 0$ (given that we begin in state (0,0) at time $a$) is $r_f C_a^b|Y_{k\ell}$, where $C_a^b$ is the total amount of time within $[a,b]$ in which $i$ and $j$ reside in the same focus.  By Lemma~\ref{l_ctime}, $C_a^b|Y_{k\ell} \to (b-a)/M$ as $r_m\to \infty$, which is invariant to $Y_{k\ell}$.  By symmetry, this result also holds (\emph{mutatis mutandis}) for the $(0,0) \to (0,1)$ transition.  

Given that the probability of an $i,j$ dyad transition for any previous $[a,b]$ interval is invariant to the time 0 state of $Y_{k\ell}$ (whatever state the dyad itself happens to be in), it follows that the state of the dyad at time 0 cannot depend on $Y_{k,\ell}$.  Thus, in the limit as $r_m\to\infty$, $Y_{ij}$ is independent of $Y_{k\ell}$.
\end{proof}

The intuition behind Theorem~\ref{t_indep} is fairly simple: edges can only influence formation events within dyads, implying that the only remaining source of dependence between edges in different dyads must arise from implicit information regarding vertex co-residence embedded in the edge structure.  However, Lemma~\ref{l_ctime} tells us that (in the fast-migration limit) this information washes out arbitrarily fast, and hence dyads must become independent as $r_m$ diverges.  This is the same mechanism that produces independence in the undirected CFP (and, though we do not show it here, it is trivially true for the directed CFP without reciprocity as well).

\subsubsection{ERGM Representation}

Since the equilibrium graph distribution from the CFPR in the fast-mixing limit is homogeneous with dyadic independence, it follows from the Hammersley-Clifford theorem that it has an ERGM representation in terms of counts of edges and mutuals \citep[see][]{pattison.robins:sm:2002}.  To find the associated parameter values from the CFPR parameters, we proceed as follows.

First, to obtain the edge parameter, we consider the conditional odds of an $(i,j)$ edge without an incoming reciprocated tie.  To simplify notation, we write edge states in terms of the realized adjacency matrix, $Y$, and introduce indicator variables for dyadic states.  Specifically, we define $\mathbb{A}_{ij}$ to be an indicator for the $i,j$ dyad being dyadic, $\mathbb{M}_{ij}$ for its being mutual, and $\mathbb{N}_{ij}$ for its being null.  With this notation, the conditional probability of an unreciprocated $(i,j)$ edge in the fast-mixing limit is seen to be
\begin{align}
\Pr( \mathbb{A}_{ij} | Y_{ji}=0) &= \frac{\Pr(Y_{ji}=0|\mathbb{A}_{ij}) \Pr(\mathbb{A}_{ij})}{\Pr(Y_{ji}=0|\mathbb{A}_{ij}) \Pr(\mathbb{A}_{ij}) + \Pr(Y_{ji}=0|\mathbb{N}_{ij}) \Pr(\mathbb{N}_{ij}) } \nonumber \\
 &= \frac{1}{2} \frac{\frac{D_a}{D}}{\frac{1}{2} \frac{D_a}{D} + 1 \frac{D_n}{D}} \nonumber \\
 &= \frac{1}{1 + 2 \frac{D_n}{D_a}}
\end{align}
and thus the conditional odds of an edge given non-reciprocation are
\begin{align}
\frac{\Pr( Y_{ij}=1 | Y_{ji}=0)}{\Pr( Y_{ij}=0 | Y_{ji}=0)} &= \frac{1}{1 + 2 \frac{D_n}{D_a}} \frac{1 + 2 \frac{D_n}{D_a}}{2 \frac{D_n}{D_a}} \nonumber \\
                   &= \frac{1}{2} \frac{D_a}{D_n} \nonumber \\
  &= \frac{1}{2} \frac{2 D }{ \frac{r_f}{r_\ell} + 2 + M \frac{r_\ell}{r_f} } \left[\frac{M D}{ \left(\frac{r_f}{r_\ell}\right)^2 + 2 \frac{r_f}{r_\ell} + M }\right]^{-1}  \nonumber \\
  &= \frac{1}{ \frac{r_f}{r_\ell} + 2 + M \frac{r_\ell}{r_f} }  \left[\frac{1}{ \left(\frac{r_f}{r_\ell}\right)^2/M + \frac{2}{M} \frac{r_f}{r_\ell} + 1 }\right]^{-1} \nonumber \\
  &= \frac{ \left(\frac{r_f}{r_\ell}\right)^2/M + \frac{2}{M} \frac{r_f}{r_\ell} + 1 }{ \frac{r_f}{r_\ell} + 2 + M \frac{r_\ell}{r_f} } \nonumber \\
  &\xrightarrow[N\to\infty]{} \frac{1}{ \frac{N}{P} \frac{r_\ell}{r_f} }  \nonumber \\
  &= \frac{P}{N} \frac{r_f}{r_\ell},
\end{align}
implying that the edge parameter in the large-$N$, fast-mixing case must be 
\begin{equation}
\theta_e = \log\left(P \frac{r_f}{r_\ell}\right) - \log N,
\end{equation}
where we are using the fact that, in ERGM form, the edge parameter in an edge/mutual model must be the logit of the conditional probability of an unreciprocated edge.

To obtain the mutuality parameter, we now consider the probability of an $(i,j)$ tie given an incoming reciprocated edge in the fast-mixing limit:
\begin{align}
\Pr( \mathbb{M}_{ij} | Y_{ji}=1) &= \frac{\Pr(Y_{ji}=1|\mathbb{M}_{ij}) \Pr(\mathbb{M}_{ij})}{\Pr(Y_{ji}=1|\mathbb{M}_{ij}) \Pr(\mathbb{M}_{ij}) + \Pr(Y_{ji}=1|\mathbb{A}_{ij}) \Pr(\mathbb{A}_{ij}) } \nonumber \\
 &= \frac{\frac{D_m}{D}}{\frac{D_m}{D} + \frac{1}{2} \frac{D_a}{D}} \nonumber \\
 &= \left[1 + \frac{1}{2} \frac{D_a}{D_m}\right]^{-1}.
\end{align}
The odds of an edge given reciprocation are hence
\begin{align}
\frac{\Pr( Y_{ij}=1 | Y_{ji}=1)}{\Pr( Y_{ij}=0 | Y_{ji}=1)} &= \frac{1}{1 + \frac{1}{2} \frac{D_a}{D_m}} \frac{1 + \frac{1}{2} \frac{D_a}{D_m}}{\frac{1}{2} \frac{D_a}{D_m}} \nonumber \\
                  &= 2 \frac{D_m}{D_a} \nonumber \\
  &= \frac{2 D}{ 1 + 2 \frac{r_\ell}{r_f} + M \left(\frac{r_\ell}{r_f}\right)^2 }  \left[\frac{2 D }{ r_f/r_\ell + 2 + M r_\ell/r_f }\right]^{-1} \nonumber \\
  &= \frac{ r_f/r_\ell + 2 + M r_\ell/r_f }{ 1 + 2 r_\ell/r_f + M (r_\ell/r_f)^2 } \nonumber \\
  &\xrightarrow[N\to\infty]{} \frac{ M r_\ell/r_f }{ M (r_\ell/r_f)^2 } \nonumber \\
  &= r_f/r_\ell,
\end{align}
and the mutuality parameter in the large-$N$, fast-mixing limit must therefore be
\begin{align}
\theta_m &= \log[r_f/r_\ell] - \theta_e  \nonumber \\
   &= \log[r_f/r_\ell] - \log[P r_f/r_\ell] + \log N \nonumber \\
   &= -\log P + \log N.
\end{align}

Putting this all together, and moving terms involving $N$ to the reference measure, we obtain the reference model

\begin{gather}
\psi_e = \log[P r_f/r_\ell]\\
\psi_m = -\log P\\
h(y) = N^{t_m(y)-t_e(y)}.
\end{gather}

It is noteworthy that using this reference model together with free parameters for edges and mutuals leads to a model family equivalent to the offset family proposed by \citet{krivitsky.kolaczyk:ss:2015} for preserving reciprocity and mean degree in large networks.  Our development thus provides a mechanistic interpretation for the offset model in terms of a contact formation process; we consider additional substantive implications below.

\subsection{Requirements for Time Scale Separation} \label{sec_timescale}

In section~\ref{sec_eq}, we obtained exact expressions for the behavior of the CFPR in the large-$N$, fast-mixing limit.  One may reasonably ask, however, how fast the migration rate must be (relative to $r_f$ and $r_\ell$) for these expressions to apply.  To investigate this, we simulate draws from the CFPR with varying migration rates, examining the properties of the resulting networks.  For our simulation study, we employ a full factorial design with treatments $N \in (50,100,200,400)$, $P \in (5,10,25)$, and $log_5 r_m \in (-4,-3,\ldots,3,4)$ (500 replicate draws per condition).  For all draws we take $r_f=1$, $r_\ell=5$, and simulate for 100 time units (the resulting draw being the final system state); simulations are initialized with homogeneous Bernoulli graphs at the limiting expected density, with random assignment of vertices to foci.  Simulations and analyses were performed using a combination of custom scripts and tools from the \texttt{sna} package for the R statistical computing system \citep{butts:jss:2008b}.

To begin, we consider the time scale required for effective convergence of mean degree and reciprocity to the asymptotic limit.  Figure~\ref{f_meandeg} shows the variation in mean degree by $P$ and $r_m$ for the largest simulated treatment ($N=400$), together with theoretical limits (dotted lines) for the slow and fast mixing regimes respectively.  These results mirror those seen for the undirected CFP in \citet{butts:jms:2019}, with excellent convergence to the fast-mixing limit when $r_m$ is roughly two orders of magnitude faster than the edge formation/loss rates.  As with the CFP case, we see little if any impact of $P$ on convergence.  Turning to reciprocity, figure~\ref{f_recip} shows mean reciprocity by $r_m$ for the full range of simulated cases.  For comparison, we also include simulations from a directed CFP \emph{without} reciprocation (i.e., the original CFP, but applied to directed rather than undirected edges).  Interestingly, we observe that all models produce similar behavior in the slow mixing case: this is because (1) migration plays no role here (and hence there is no excess edge formation due to closure of dyads that span foci), and (2) all scenarios employed here have the same expected density in the slow mixing limit.  The difference between the CFPR and the CFP becomes evident when the mixing rate increases.  Under fast mixing, the reciprocation rate under the CFP falls because vertices are less frequently co-located with their neighbors.  Under the CFPR, however, edgewise reciprocity holds constant.  Note that, unlike mean degree, reciprocity preservation holds at \emph{all} mixing rates: this is because the CFPR acts to nullify the impact of foci on reciprocation opportunity, decoupling the migration process from reciprocation altogether.  (Compare with the use of the mutuality offset to nullify the effect of the edge offset in the development of \citet{krivitsky.kolaczyk:ss:2015}.)

\begin{figure}[h]
  \centering
    \includegraphics[width=0.9\textwidth]{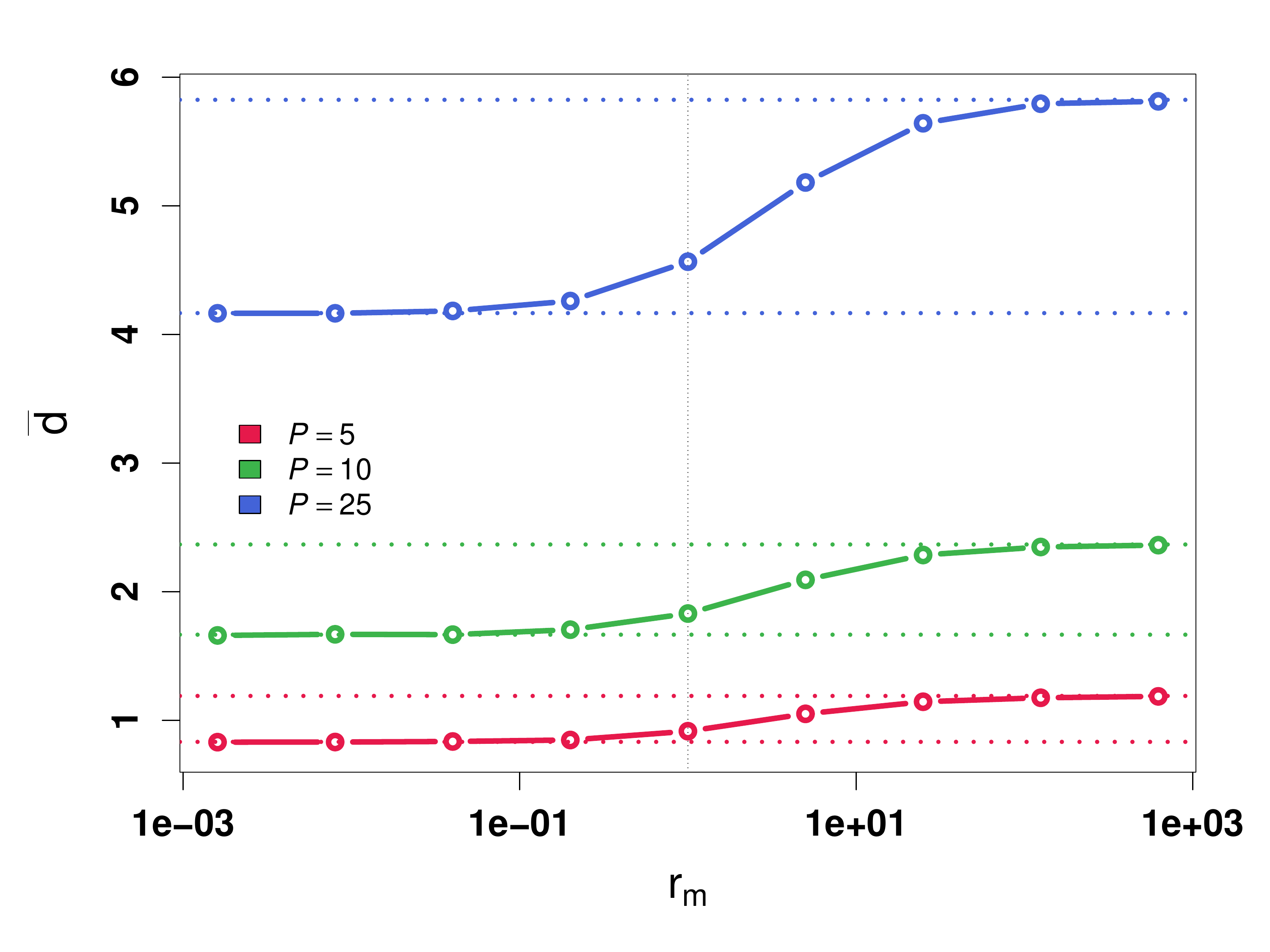}
    \caption{Simulated mean degree (heavy lines) and theoretical limits (dotted lines) as a function of $r_m$ and $P$; 95\% confidence intervals for the simulation mean are too narrow to be visible.  Effective time scale separation is achieved by $r_m>100$.  \label{f_meandeg}}
\end{figure}

\begin{figure}[h]
  \centering
    \includegraphics[width=0.9\textwidth]{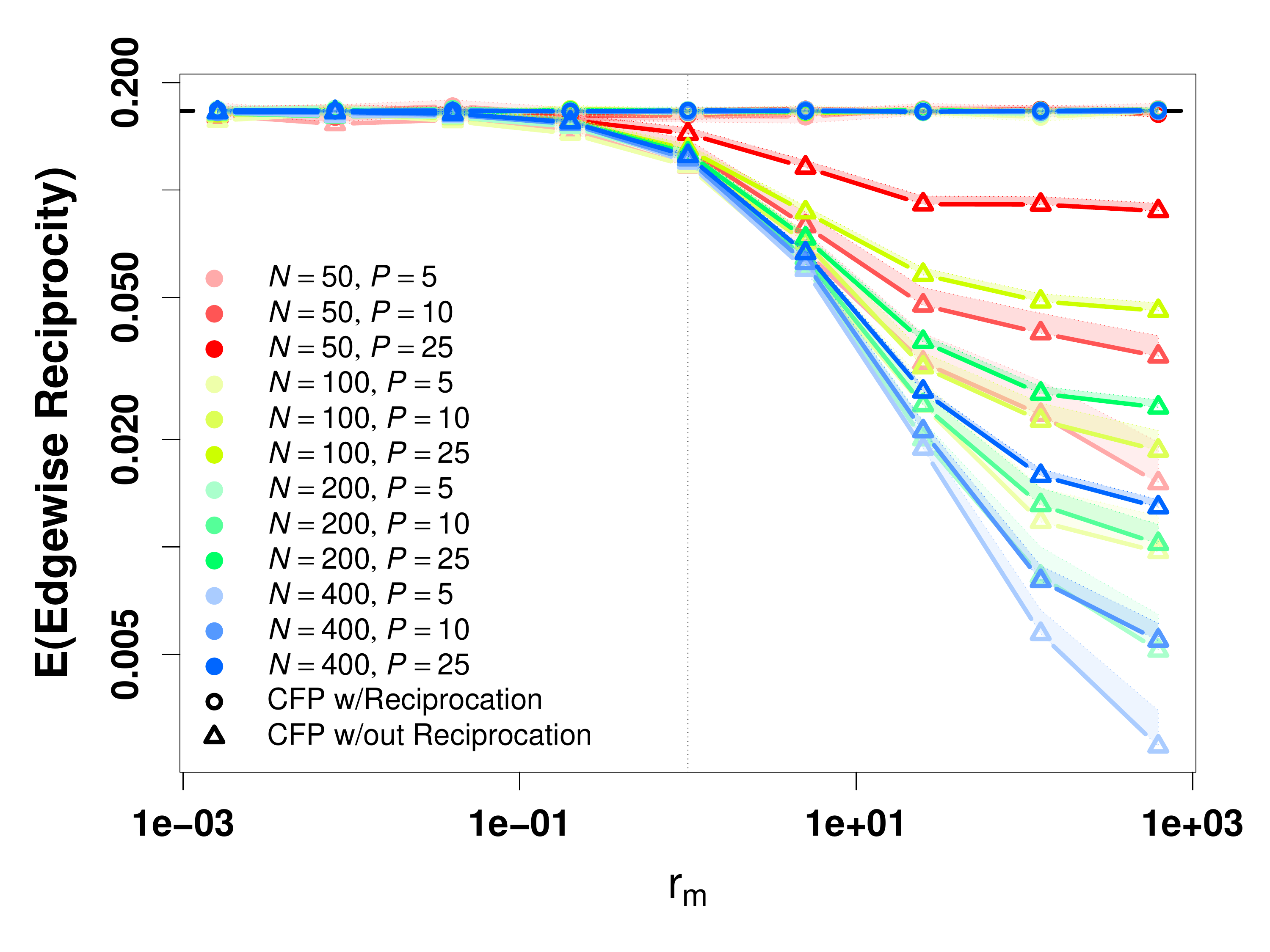}
    \caption{Simulated mean edgewise reciprocity by mixing rate, for different $N,P$ conditions (color) and process types (point shape); shaded areas indicate 95\% confidence intervals.  At low mixing rates, reciprocity rates for all models correspond to the common limiting density.  Under fast mixing, reciprocity falls in the regular CFP conditions (triangles) while remaining constant under all CFPR conditions (circles).\label{f_recip}}
\end{figure}

Beyond proper scaling of mean degree and reciprocity, we also wish to determine the migration rate needed for the approximation of dyadic independence to hold.  To assess this, we consider the difference between the expected triad census obtained under the CFPR (and, for comparison, directed CFP) under each condition and the census that would be obtained under a corresponding $u|man$ model with the same dyad frequencies.  For this purpose we take the mean triad census over all replications in each condition, repeating this process for 5,000 draws from the $u|man$ distribution with expected dyad census equal to the mean dyad census from the contact process simulations.  We then calculate Hotelling's $T^2$ (a multivariate generalization of the $t$ statistic) for the contact process triad census versus the corresponding $u|man$ census in each condition.  Figure~\ref{f_triad} shows the resulting $T^2$ values by condition and model, as a function of migration rate.  Unsurprisingly, the presence of latent foci introduce clustering at low $r_m$, resulting in a triad census that deviates substantially from what would be obtained under dyadic independence.  As the mixing rate increases, however, dyadic independence weakens, and the triad census converges to the $u|man$ limit.  In particular, once $r_m$ is approximately two orders of magnitude greater than the formation and dissolution rates, we see little to no significant deviation from dyadic independence.  We also note that the directed CFP and the CFPR behave quite similarly, with no obvious impact of adding reciprocation opportunities to the base CFP on the convergence rate.  By contrast, higher values of $N$ and $P$ are associated with somewhat slower convergence at low $r_m$, though all conditions seem to collapse rapidly to the independent limit by around $r_m=100$.

\begin{figure}[h]
  \centering
    \includegraphics[width=0.75\textwidth]{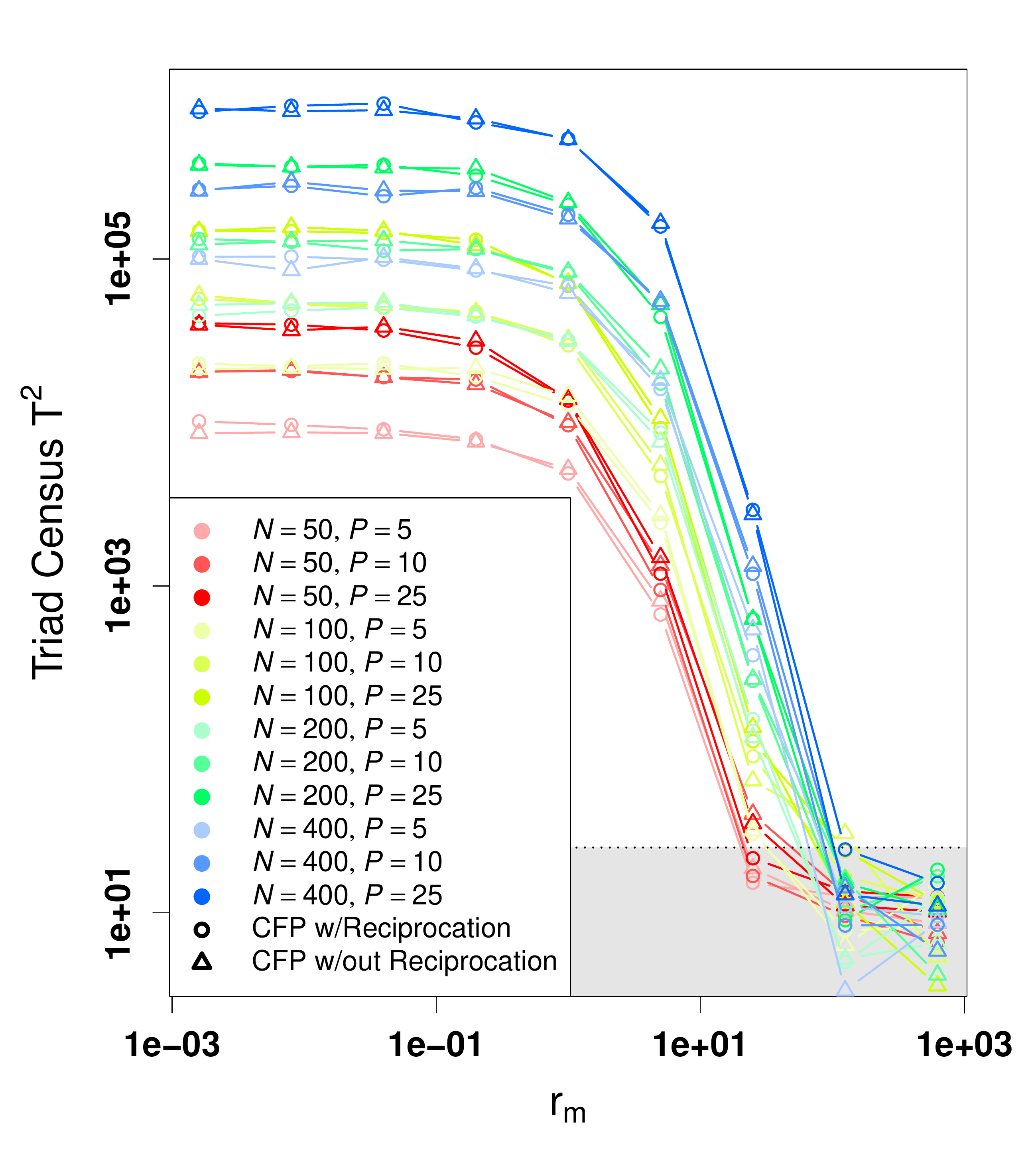}
    \caption{Convergence of the triad census to expectations under dyadic independence, as a function of $r_m$.  Lines show Hotelling $T^2$ values for the triad census under simulated networks in each condition versus a corresponding $u|man$ sample; colors indicate $C,P$ condition, with process type indicated by point shape.  Shaded area and horizontal line shows threshold for significance at the 0.05 level.  Deviations from the independent dyad model are minimal beyond $r_m>100$ for both CFP and CFPR. \label{f_triad}}
\end{figure}



\section{``Pure'' Focus Parameterization} \label{sec_varm}

We have generally assumed that, as population increases, the number of foci occupied also increases so as to maintain constant population density (i.e., constant $P$).  This is a plausible assumption in many cases, but is not essential; alternative assumptions regarding the scaling of focus count with $N$ can be employed, resulting in different mean-degree scaling.  For this purpose, it may be more useful to work with an alternative specification of the reference model that removes $P$ entirely, leaving us with a formulation in terms of $M$ alone:
\begin{gather}
\psi_e = \log[r_f/r_\ell] \nonumber \\
\psi_m = 0 \nonumber \\
h(y) = M^{t_m(y)-t_e(y)}. \label{e_altmref}
\end{gather}
This formulation is convenient in providing a helpful (if somewhat oversimplified) intuition connecting the CFPR migration process to the resulting baseline graph distribution, and in highlighting the difference between the CFP and the CFPR.  Intuitively, a basic feature of the CFP is that tie formation requires both parties to be within the same focus.  In equilibrium, the chance of finding two vertices in the same focus is $1/M$.  For a graph with $t_e(y)$ edges, this must obviously have happened $t_e(y)$ times; if we naively treat these as independent events (appealing to the assumption of fast mixing) then this suggests a reference measure that scales as $M^{-t_e(y)}$.  This is indeed a valid expression for the reference measure arising from a pure CFP, and the intuition carries.  In the case of the CFPR, this requirement of co-presence is released for reciprocating edges, of which there are $t_m(y)$.  Thus, the same intuition would suggest that we simply deduct this number from the number of required co-incidence events, leading to a reference measure that scales as $M^{t_(y)-t_e(y)}$.  This is precisely what we obtain in equation~\ref{e_altmref}, and the intuitive answer is again correct.  As we have seen, a rigorous development of the reference measure is considerably more involved than this naive intuition would suggest, and such arguments cannot be relied upon \emph{ex ante} to obtain expressions for model behavior.  However, the core insight of the intuitive argument regarding the impact of foci on edge probability in the CFP and CFPR is correct, and it may hence be helpful as an aid to interpretation.

\subsection{$M$ Scaling for Spatial Systems}

Note that, as $M$ is not in general known in empirical settings, some functional form for it in terms of $N$ must be assumed (e.g., the linear form used elsewhere in this paper).  In some settings, there may be a natural choice in this regard.  For instance, consider a spatial system with volume $V$ and ``concentration'' (i.e., spatial population density) $C$, with some $v$ being the volume of the largest volume element over which ties can be formed.  In this case, it is obvious that $N=CV$ and $M=V/v$, giving us 
\begin{align*}
h(y) &= (V/v)^{t_m(y)-t_e(y)}.
\end{align*}
This development makes no particular assumptions about dimensionality, but we can easily generalize this to e.g. a $d$-dimensional hypercube of length $L$, with $l$ being the side length of the critical volume element.  Then $V=L^d$, $v=l^d$, and
\begin{align*}
h(y) &= (L/l)^{d(t_m(y)-t_e(y))}.
\end{align*}
From this it becomes possible to predict changes in network structure arising from changing the range over which vertices can interact, the scale of the system as a whole, or even the dimensionality of the system (possibly relevant in the context of Blau spaces).  Similar developments are possible for other classes of systems.

\subsection{$M$ Scaling for Non-constant Mean Degree}

\citet{butts.almquist:jms:2015} provide a phenomenological treatment of correction for mean degree scaling of the form $\bar{d}\propto N^\gamma$, which takes the Bernoulli baseline ($\gamma=1$) and Krivitsky reference ($\gamma=0$) as special cases while also accommodating phenomena such as so-called\footnote{Contrary to what the name implies, networks undergoing ``power law densification'' are actually growing more sparse.  They are, however, doing so more slowly than they would under constant mean degree.} ``power law densification'' \citep{leskovec.et.al:kdd:2007}.  We can derive the equivalent mechanistic conditions required for this phenomenon under the CFPR by exploiting equation~\ref{e_meandeg}:
\begin{align*}
N^\gamma &\propto P \frac{r_f}{r_l}\left(\frac{r_f}{r_l}+1\right)\\ 
         &= \frac{N}{M} \frac{r_f}{r_l}\left(\frac{r_f}{r_l}+1\right).
\end{align*}
Choosing $M=N^{1-\gamma}$ leads to 
\[
\lim_{N\to\infty} \bar{d} = N^\gamma \frac{r_f}{r_l}\left(\frac{r_f}{r_l}+1\right),
\]
as desired.  We can also verify from this that the Bernoulli baseline (counting measure) implies $M=1$ (hence, no migration process) and the constant mean degree scaling implies $M \propto N$ (as has been our focus in this paper).  More generally, focus counts that scale sublinearly (e.g., $M \propto N^\alpha$, for $\alpha \in (0,1)$) will induce ``power law densificiation'' in the sense of \citep{leskovec.et.al:kdd:2007}.  We also observe from this result that settings in which the number of foci grows supralinearly in population size would be expected to result in ``super-sparse'' networks for which mean degree itself declines in $N$.  The existence of such settings is an interesting empirical question.




\section{Conclusion}

In directed relations, we often see edge reciprocation rates that remain constant or nearly constant as $N$ increases, even while density falls.  Here, we have provided a simple extension of the contact formation process that provides a mechanistic account of the phenomenon while also leading to a well-characterized family of graph distributions that can serve as the starting point for more complex model building.  As with the original CFP, we find that migration rates roughly two orders of magnitude faster than tie formation/loss rates are sufficient for timescale separation.  For social ties with durations on the order of months, this implies CFP foci with typical residence times on the order of hours or less (e.g., meeting or gathering places, discussion settings or small group interactions, etc.).  For longer-term ties with durations on the order of years, correspondingly longer residence times (e.g., on the order of days or weeks) become feasible.  These time periods are long enough to cover a wide range of transient settings (both physical and virtual) in which individuals meet, interact, and potentially initiate social relationships.

Although constant mean degree scaling is often a good starting assumption for social networks, other options are also possible.  This is easily explored under the CFP and CFPR by relaxing the convenient assumption that $P$ is constant in $N$.  Here, we noted how such relaxations can be used to incorporate \emph{a priori} theories about how $M$ scales with $N$, or alternately to adjust for different types of mean degree scaling.  Further work on factors that might be expected to alter the richness of the social ecology (and hence $M$) could be helpful in suggesting additional regimes for further exploration.

Finally, we note that the specific mechanism invoked in the CFPR to explain constant reciprocation - that a nomination by alter always leaves alter an available target for nomination by ego - is not the only means by which reciprocity could be sustained, nor is it necessarily active in all settings.  For instance, in a ``blind'' nomination context in which ego cannot necessarily be assumed to know who nominates him or her, it is more natural to assume a CFP baseline than a CFPR baseline.  Contact formation is only one of many types of micro-level processes that can lead to social structure, and development of alternative models would give us a richer set of options for explaining the regularities of social and other networks.

\bibliography{ctb}

\begin{thebibliography}{}

\bibitem[Blau, 1972]{blau:ssr:1972}
Blau, P.~M. (1972).
\newblock Interdependence and hierarchy in organizations.
\newblock {\em Social Science Research}, 1:1--24.

\bibitem[Blau, 1977]{blau:ajs:1977}
Blau, P.~M. (1977).
\newblock A macrosocial theory of social structure.
\newblock {\em American Journal of Sociology}, 83(1):26--54.

\bibitem[Butts, 2008]{butts:jss:2008b}
Butts, C.~T. (2008).
\newblock Social network analysis with sna.
\newblock {\em Journal of Statistical Software}, 24(6).

\bibitem[Butts, 2015]{butts:jms:2015}
Butts, C.~T. (2015).
\newblock A novel simulation method for binary discrete exponential families,
  with application to social networks.
\newblock {\em Journal of Mathematical Sociology}, 39(3):174--202.

\bibitem[Butts, 2018]{butts:jms:2018}
Butts, C.~T. (2018).
\newblock A perfect sampling method for exponential family random graph models.
\newblock {\em Journal of Mathematical Sociology}, 42(1):17--36.

\bibitem[Butts, 2019]{butts:jms:2019}
Butts, C.~T. (2019).
\newblock A dynamic process interpretation of the sparse {ERGM} reference
  model.
\newblock {\em Journal of Mathematical Sociology}, 43(1):40--57.

\bibitem[Butts et~al., 2012]{butts.et.al:sn:2012}
Butts, C.~T., Acton, R.~M., Hipp, J.~R., and Nagle, N.~N. (2012).
\newblock Geographical variability and network structure.
\newblock {\em Social Networks}, 34:82--100.

\bibitem[Butts and Almquist, 2015]{butts.almquist:jms:2015}
Butts, C.~T. and Almquist, Z.~W. (2015).
\newblock A flexible parameterization for baseline mean degree in
  multiple-network {ERGM}s.
\newblock {\em Journal of Mathematical Sociology}, 39(3):163--167.

\bibitem[Butts et~al., 2007]{butts.et.al:jms:2007}
Butts, C.~T., Petrescu-Prahova, M., and Cross, B.~R. (2007).
\newblock Responder communication networks in the {W}orld {T}rade {C}enter
  {D}isaster: Implications for modeling of communication within emergency
  settings.
\newblock {\em Journal of Mathematical Sociology}, 31(2):121--147.

\bibitem[Carley, 1991]{carley:asr:1991}
Carley, K.~M. (1991).
\newblock A theory of group stability.
\newblock {\em American Sociological Review}, 56(3):331--354.

\bibitem[Caselli et~al., 2006]{caselli.et.al:bk:2006}
Caselli, G., Vallin, J., and Wunsch, G. (2006).
\newblock {\em Demography: Analysis and Synethesis}.
\newblock Elsevier, Amsterdam.

\bibitem[Dunbar, 1997]{dunbar:bk:1997}
Dunbar, R. (1997).
\newblock {\em Grooming, Gossip, and the Evolution of Language}.
\newblock Harvard University Press, Cambridge, MA.

\bibitem[Feld, 1981]{feld:ajs:1981}
Feld, S. (1981).
\newblock The focused organization of social ties.
\newblock {\em American Journal of Sociology}, 1986:1015--1035.

\bibitem[Galbraith, 1977]{galbraith:bk:1977}
Galbraith, J. (1977).
\newblock {\em Organization Design}.
\newblock Addison-Wesley, Reading, MA.

\bibitem[Goodreau et~al., 2009]{goodreau.et.al:d:2009}
Goodreau, S.~M., Kitts, J.~A., and Morris, M. (2009).
\newblock Birds of a feather, or friend of a friend?: Using exponential random
  graph models to investigate adolescent social networks.
\newblock {\em Demography}, 46(1):103--125.

\bibitem[Grazioli et~al., 2019]{grazioli.et.al:jpcB:2019}
Grazioli, G., Yu, Y., Unhelkar, M.~H., Martin, R.~W., and Butts, C.~T. (2019).
\newblock Network-based classification and modeling of amyloid fibrils.
\newblock {\em Journal of Physical Chemistry, B}, 123(26):5452--5462.

\bibitem[Harris et~al., 2009]{harris.et.al:web:2009}
Harris, K.~M., Halpern, C.~T., Whitsel, E., Hussey, J., Tabor, J., Entzel, P.,
  and Udry, J.~R. (2009).
\newblock The {N}ational {L}ongitudinal {S}tudy of {A}dolescent {H}ealth:
  Research design.
\newblock URL: http://www.cpc.unc.edu/projects/addhealth/design.

\bibitem[Huitsing et~al., 2012]{huitsing.et.al:sn:2012}
Huitsing, G., van Duijn, M.~A., Snijders, T.~A., Wang, P., Sainio, M.,
  Salmivalli, C., and Veenstra, R. (2012).
\newblock Univariate and multivariate models of positive and negative networks:
  Liking, disliking, and bully--victim relationships.
\newblock {\em Social Networks}, 34(4):645--657.

\bibitem[Hunter and Handcock, 2006]{hunter.handcock:jcgs:2006}
Hunter, D.~R. and Handcock, M.~S. (2006).
\newblock Inference in curved exponential family models for networks.
\newblock {\em Journal of Computational and Graphical Statistics}, 15:565--583.

\bibitem[Hunter et~al., 2012]{hunter.et.al:jcgs:2012}
Hunter, D.~R., Krivitsky, P.~N., and Schweinberger, M. (2012).
\newblock Computational statistical methods for social network analysis.
\newblock {\em Journal of Computational and Graphical Statistics}, 21:856--882.

\bibitem[Jaynes, 1983]{jaynes:bk:1983}
Jaynes, E.~T. (1983).
\newblock {\em Papers on Probability, Statistics, and Statistical Physics}.
\newblock Dordrecht, Reidel.
\newblock Rosencrantz, R. D. (Ed.).

\bibitem[Krivitsky et~al., 2011]{krivitsky.et.al:statm:2011}
Krivitsky, P.~N., Handcock, M.~S., and Morris, M. (2011).
\newblock Adjusting for network size and composition effects in
  exponential-family random graph models.
\newblock {\em Statistical Methodology}, 8(4):319--339.

\bibitem[Krivitsky and Kolaczyk, 2015]{krivitsky.kolaczyk:ss:2015}
Krivitsky, P.~N. and Kolaczyk, E.~D. (2015).
\newblock On the question of effective sample size in network modeling: An
  asymptotic inquiry.
\newblock {\em Statistical Science}, 30:184--198.

\bibitem[Leskovec et~al., 2007]{leskovec.et.al:kdd:2007}
Leskovec, J., Kleinberg, J., and Faloutsos, C. (2007).
\newblock Graph evolution: Densification and shrinking diameters.
\newblock {\em ACM Transactions on Knowledge Discovery from Data}, 1(1).

\bibitem[Mayhew, 1984]{mayhew:jms:1984a}
Mayhew, B.~H. (1984).
\newblock Baseline models of sociological phenomena.
\newblock {\em Journal of Mathematical Sociology}, 9:259--281.

\bibitem[Mayhew and Levinger, 1976]{mayhew.levinger:ajs:1976}
Mayhew, B.~H. and Levinger, R.~L. (1976).
\newblock Size and density of interaction in human aggregates.
\newblock {\em American Journal of Sociology}, 82:86--110.

\bibitem[Mayhew et~al., 1995]{mayhew.et.al:sf:1995}
Mayhew, B.~H., McPherson, J.~M., Rotolo, M., and Smith-Lovin, L. (1995).
\newblock Sex and race homogeneity in naturally occurring groups.
\newblock {\em Social Forces}, 74(1):15--52.

\bibitem[McFarland et~al., 2014]{mcfarland.et.al:asr:2014}
McFarland, D.~A., Moody, J., Diehl, D., Smith, J.~A., and Thomas, R.~J. (2014).
\newblock Network ecology and adolescent social structure.
\newblock {\em American Sociological Review}, 79(6):1088--1121.

\bibitem[McPherson, 2004]{mcpherson:icc:2004}
McPherson, J.~M. (2004).
\newblock A {B}lau space primer: Prolegomenon to an ecology of affiliations.
\newblock {\em Industrial and Corporate Change}, 13:263--280.

\bibitem[McPherson et~al., 2001]{mcpherson.et.al:ars:2001}
McPherson, J.~M., Smith-Lovin, L., and Cook, J.~M. (2001).
\newblock Birds of a feather: Homophily in social networks.
\newblock {\em Annual Review of Sociology}, 27:415--444.

\bibitem[Morris et~al., 2008]{morris.et.al:jss:2008}
Morris, M., Handcock, M.~S., and Hunter, D.~R. (2008).
\newblock Specification of exponential-family random graph models: Terms and
  computational aspects.
\newblock {\em Journal of Statistical Software}, 24(4):1--24.

\bibitem[Pattison and Robins, 2002]{pattison.robins:sm:2002}
Pattison, P.~E. and Robins, G.~L. (2002).
\newblock Neighborhood-based models for social networks.
\newblock {\em Sociological Methodology}, 32:301--337.

\bibitem[Schweinberger et~al., 2019]{schweinberger.et.al:ss:2019}
Schweinberger, M., Krivitsky, P.~N., Butts, C.~T., and Stewart, J. (2019).
\newblock Exponential-family models of random graphs: Inference in finite-,
  super-, and infinite-population scenarios.
\newblock {\em Statistical Science}, forthcoming.

\bibitem[Schweinberger and Stewart, 2019]{schweinberger.stewart:as:2019}
Schweinberger, M. and Stewart, J. (2019).
\newblock Concentration and consistency results for canonical and curved
  exponential-family models of random graphs.
\newblock {\em The Annals of Statistics}, forthcoming.

\bibitem[Skvoretz et~al., 2004]{skvoretz.et.al:sn:2004}
Skvoretz, J., Fararo, T.~J., and Agneessens, F. (2004).
\newblock Advances in biased net theory: Definitions, derivations, and
  estimations.
\newblock {\em Social Networks}, 26:113--139.

\bibitem[Snijders, 2001]{snijders:sm:2001}
Snijders, T. A.~B. (2001).
\newblock The statistical evaluation of social network dynamics.
\newblock {\em Sociological Methodology}, 31:361--395.

\bibitem[Snijders, 2002]{snijders:joss:2002}
Snijders, T. A.~B. (2002).
\newblock {M}arkov chain {M}onte {C}arlo estimation of exponential random graph
  models.
\newblock {\em Journal of Social Structure}, 3(2).

\bibitem[Van~Duijn et~al., 1999]{vanduijn.et.al:sn:1999}
Van~Duijn, M.~A., Van~Busschbach, J.~T., and Snijders, T.~A. (1999).
\newblock Multilevel analysis of personal networks as dependent variables.
\newblock {\em Social Networks}, 21(2):187--210.

\end{thebibliography}


\end{document}